\documentclass[submission,copyright,creativecommons]{eptcs}
\providecommand{\event}{Non-Classical Logics 2024 (NCL'24)}

\usepackage{iftex}

\ifpdf
  \usepackage{underscore}         
  \usepackage[T1]{fontenc}        
\else
  \usepackage{breakurl}           
\fi

\usepackage{color}
\usepackage{amsmath, amssymb}
\usepackage{amsthm}

\theoremstyle{definition}
\newtheorem{theorem}{Theorem} 
\newtheorem{definition}{Definition}
\newtheorem{proposition}[definition]{Proposition}
\newtheorem{lemma}[definition]{Lemma}
\newtheorem{corollary}[definition]{Corollary}
\newtheorem{remark}[definition]{Remark}

\newcommand{\Not}{{\sim}}

\newcommand{\dto}{{\to_d}}
\newcommand{\rdland}{{\land_d^r}}
\newcommand{\dequiv}{{\leftrightarrow_d}}

\title{The Disjunction-Free Fragment of {\bf D$_2$} is Three-Valued\thanks{The main result was presented with a rather different narrative at \emph{Logic in Bochum 2}, \emph{CCPEA 2016} in Seoul, \emph{Paradoxes, Logic and Philosophy} in Beijing, \emph{V Workshop on Philosophical of Logic} in Buenos Aires, \emph{Prague Seminar on Paraconsistent Logic}, a colloquium in Munich and \emph{ISRALOG17} in Haifa. I owe a special debt of gratitude to Dave Ripley whose comments led me to rethink the overall presentation of the main result.  An earlier version of this article was presented at: \emph{Non-classical modalities} in Mexico City, \emph{the Eleventh Smirnov Readings in Logic} in Moscow, \emph{CoPS-FaM-2019} in Gda\'nsk, \emph{Paris-Bochum-Moscow Workshop in Mathematical Philosophy} in Paris and another colloquium in Munich. I would like to thank the organizers of these events for their kind invitations, warm hospitality and helpful discussions, as well as the audiences at these meeting for useful comments. I would also like to thank Jonas Rafael Becker Arenhart and Fabio De Martin Polo for helpful discussions and comments. Finally, but not the least, I would like to thank the referees for their very kind, detailed, and supportive comments that improved the presentation of the paper. The preparation of an earlier version of this article was supported by a Sofja Kovalevskaja Award of the Alexander von Humboldt-Foundation, funded by the German Ministry for Education and Research.}}
\author{
Hitoshi Omori
\institute{Graduate School of Information Sciences\\
Tohoku University\\
Sendai, Japan}
\email{hitoshiomori@gmail.com}
}
\def\titlerunning{The disjunction-free fragment of {\bf D$_2$} is three-valued}
\def\authorrunning{H. Omori}
\begin{document}
\maketitle

\begin{abstract}
In this article, the disjunction-free fragment of Ja\'skowski's discussive logic {\bf D$_2$} in the language of classical logic is shown to be complete with respect to three- and four-valued semantics. As a byproduct, a rather simple axiomatization of the disjunction-free fragment of {\bf D$_2$} is obtained. Some implications of this result are also discussed.
\end{abstract}

\section{Introduction}

Stanis{\l}aw Ja\'skowski is known to be one of the modern founders of paraconsistent logic, together with Newton C. A. da Costa. The most important contribution of Ja\'skowski is that he clearly distinguished two notions for a theory, namely a theory being \emph{contradictory} (or \emph{inconsistent} in \cite{J1}) and a theory being \emph{trivial} (or \emph{overfilled} in \cite{J1}). In addition to this distinction, he also presented a system of paraconsistent logic known as {\bf D$_2$} which is often referred to as discursive logic or discussive logic (cf. \cite{J1, J2}). 

In this article, the disjunction-free fragment of Ja\'skowski's discussive logic is shown to be complete with respect to three- and four-valued semantics. Note here that {\bf D$_2$} is known to be not complete with respect to any finitely many-valued semantics, which is proved by Jerzy Kotas in \cite{Kotas1974b}. 
As a byproduct of the main result, a simple axiomatization of the disjunction-free fragment of Ja\'skowski's discussive logic in the language of classical logic is obtained. For the problem of axiomatization of {\bf D$_2$}, see \cite{OAD2}. 

\section{Semantics and proof theory}

The propositional languages in this article consist of a finite set $\mathsf{S}$ of propositional connectives and a countable set $\mathsf{Prop}$ of propositional variables. The languages are referred to as $\mathcal{L}$, $\mathcal{L}_r^-$, $\mathcal{L}_r$, $\mathcal{L}_l^-$ and $\mathcal{L}_l$ when $\mathsf{S}$ are $\{ \Not , \dto , \land , \lor \}$, $\{ \Not , \dto , \land_d^r \}$, $\{ \Not , \dto , \land_d^r , \lor \}$, $\{ \Not , \dto , \land_d^l \}$,  and $\{ \Not , \dto , \land_d^l , \lor \}$, respectively. Note that the languages $\mathcal{L}$ and $\mathcal{L}_r$ were introduced by Ja\'skowski in \cite{J1} and \cite{J2}, respectively.\footnote{As correctly pointed out by a referee, Ja\'skowski also included the discussive biconditional as a \emph{primitive} connective. However, in view of \cite[Proposition 1]{OAD2}, I will treat the discussive biconditional as a \emph{defined} connective.} The language $\mathcal{L}_l$ has been considered in a number of papers including \cite{dacosta1977new,Vasyukov}. The language $\mathcal{L}_r^-$ is the main one dealt with in this paper, but I will also refer to the other languages when it is helpful.\footnote{My emphasis on the languages $\mathcal{L}_r^-$ and $\mathcal{L}_r$ is a personal choice paying my respect to Ja\'skowski for introducing the first discussive conjunction in \cite{J2}. However, the main observation of the paper carries over for other languages, and some of the details are spelled out in \S\ref{subsec:left-discussive-conjunction} and \S\ref{subsec:non-discussive-conjunction}.} The set of formulas defined as usual in $\mathcal{L}$, $\mathcal{L}_r^-$ and $\mathcal{L}_l^-$, are denoted by $\mathsf{Form}$, $\mathsf{Form}_r^-$ and $\mathsf{Form}_l^-$, respectively. Moreover, a formula is denoted by $A$, $B$, $C$, etc. and a set of formulas by $\Gamma$, $\Delta$, $\Sigma$, etc.

\subsection{Semantics for the disjunction-free fragment of {\bf D$_2$}}

The original semantics of Ja\'skowski can be precisified by making use of translations into modal language, but here I follow Janusz Ciuciura (cf. \cite{ciuciura2008frontiers}) who stated the semantics without the help of translation.

\begin{definition}[{\bf D$_2^-$}-model]
{\bf D$_2^-$}-model for $\mathcal{L}_r^-$ is a pair $\langle W , v \rangle$ where $W$ is a non-empty set and $v: W\times \mathsf{Prop} \longrightarrow \{ 0 , 1 \}$, an assignment of truth values to state-variable pairs. Valuations $v$ are then extended to interpretations $I$ to state-formula pairs by the following conditions.
\begin{itemize}
\setlength{\parskip}{0cm}
\setlength{\itemsep}{0cm}
\item $I(w, p)=v(w, p)$, for all $w\in W$ and for all $p\in\mathsf{Prop}$;
\item $I(w, \Not A)=1$ iff  $I(w, A)=0$; 
\item $I(w, A \rdland B)=1$ iff $I(w, A)=1 \text{ and for some } x\in W (I(x, B)=1)$; 
\item $I(w, A \dto B)=1$ iff for all $x\in W (I(x, A)=0) \text{ or } I(w, B)=1$.
\end{itemize}
Furthermore, $\Gamma\models_d A$ iff for every {\bf D$_2^-$}-model $\langle W , v \rangle$, if for all $B\in \Gamma$, there is $x\in W$ such that $I(x, B)=1$, then $I(y, A)=1$ for some $y\in W$.
\end{definition}

\begin{remark}
Note that the semantic consequence relation is defined in an unusual way, which is not a mistake, but a definition that reflects the original idea of Ja\'skowski. 
\end{remark}

Now, by considering a special case of the Kripke semantics in which the cardinality of $W$ is two, the following four-valued semantics is obtained.

\begin{definition}\label{def:four-valuedD2-}
A \emph{four-valued {\bf D$_2^-$}-interpretation} of $\mathcal{L}_r^-$ is a function $v: \mathsf{Prop}\longrightarrow \{ \mathbf{1}, \mathbf{i}, \mathbf{j}, \mathbf{0} \}$. Given a four-valued {\bf D$_2^-$}-interpretation $v$, this is extended to a function $I$ that assigns every formula a truth value by truth functions depicted in the form of truth tables as follows:
\begin{displaymath}
{\small
\begin{tabular}{c|c}
$A$  &  $\Not A$\\ \hline
$\mathbf{1}$ &  $\mathbf{0}$\\
$\mathbf{i}$ & $\mathbf{j}$\\
$\mathbf{j}$ & $\mathbf{i}$\\
$\mathbf{0}$ &  $\mathbf{1}$\\
\end{tabular}
\quad
\begin{tabular}{c|cccc}
$A \rdland B$ & $\mathbf{1}$ & $\mathbf{i}$ & $\mathbf{j}$ & $\mathbf{0}$ \\ \hline
$\mathbf{1}$ & $\mathbf{1}$ & $\mathbf{1}$ & $\mathbf{1}$ & $\mathbf{0}$\\
$\mathbf{i}$ & $\mathbf{i}$ & $\mathbf{i}$ & $\mathbf{i}$ & $\mathbf{0}$\\
$\mathbf{j}$ & $\mathbf{j}$ & $\mathbf{j}$ & $\mathbf{j}$ & $\mathbf{0}$\\
$\mathbf{0}$ & $\mathbf{0}$ & $\mathbf{0}$ & $\mathbf{0}$ & $\mathbf{0}$\\
\end{tabular}
\quad
\begin{tabular}{c|cccc}
$A \dto B$ & $\mathbf{1}$ & $\mathbf{i}$ & $\mathbf{j}$ & $\mathbf{0}$ \\ \hline
$\mathbf{1}$ & $\mathbf{1}$ & $\mathbf{i}$ & $\mathbf{j}$ & $\mathbf{0}$\\
$\mathbf{i}$ & $\mathbf{1}$ & $\mathbf{i}$ & $\mathbf{j}$ & $\mathbf{0}$\\
$\mathbf{j}$ & $\mathbf{1}$ & $\mathbf{i}$ & $\mathbf{j}$ & $\mathbf{0}$\\
$\mathbf{0}$ & $\mathbf{1}$ & $\mathbf{1}$ & $\mathbf{1}$ & $\mathbf{1}$
\end{tabular}
}
\end{displaymath}
\noindent Note that the set of designated values, denoted by $\mathcal{D}_4$, is $\{ \mathbf{1}, \mathbf{i}, \mathbf{j} \}$. The semantic consequence relation $\models_4^-$ is defined in terms of preservation of designated values.
\end{definition}

\begin{remark}\label{rem:special-case}
Assume that $W{=}\{ w_1, w_2 \}$. Then, 
\begin{itemize}
\setlength{\parskip}{0cm}
\setlength{\itemsep}{0cm}
\item $v(A)=\mathbf{1}$ corresponds to $v(w_1, A)=1$ and $v(w_2, A)=1$,
\item $v(A)=\mathbf{i}$ corresponds to $v(w_1, A)=1$ and $v(w_2, A)=0$,
\item $v(A)=\mathbf{j}$ corresponds to $v(w_1, A)=0$ and $v(w_2, A)=1$,
\item $v(A)=\mathbf{0}$ corresponds to $v(w_1, A)=0$ and $v(w_2, A)=0$.
\end{itemize}
Note also that the unusual definition of the semantic consequence relation is here reflected as having three designated values. 
\end{remark}

In the above semantics, the intermediate values are representing the two possibilities depending on which of the two states or worlds falsifies the sentence. In fact, these two possibilities can be ``merged'', and the third value can stand for the case in which the two states or worlds disagree. As a result, the following three-valued semantics is obtained.

\begin{definition}
A \emph{three-valued {\bf D$_2^-$}-interpretation} of $\mathcal{L}_r^-$ is a function $v: \mathsf{Prop}\longrightarrow \{ \mathbf{1}, \mathbf{i}, \mathbf{0} \}$. Given a three-valued {\bf D$_2^-$}-interpretation $v$, this is extended to a function $I$ that assigns every formula a truth value by truth functions depicted in the form of truth tables as follows:
\begin{displaymath}
{\small
\begin{tabular}{c|c}
$A$  &  $\Not A$\\
\hline
$\mathbf{1}$ &  $\mathbf{0}$\\
$\mathbf{i}$ & $\mathbf{i}$\\
$\mathbf{0}$  &  $\mathbf{1}$\\
\end{tabular}
\quad
\begin{tabular}{c|cccc}
$A \rdland B$ & $\mathbf{1}$ & $\mathbf{i}$ &  $\mathbf{0}$ \\ \hline
$\mathbf{1}$ & $\mathbf{1}$ & $\mathbf{1}$ & $\mathbf{0}$\\
$\mathbf{i}$ & $\mathbf{i}$ & $\mathbf{i}$ & $\mathbf{0}$\\
$\mathbf{0}$ & $\mathbf{0}$ & $\mathbf{0}$ &  $\mathbf{0}$\\
\end{tabular}
\quad
\begin{tabular}{c|cccc}
$A \dto B$ & $\mathbf{1}$ & $\mathbf{i}$  & $\mathbf{0}$ \\ \hline
$\mathbf{1}$ & $\mathbf{1}$ & $\mathbf{i}$ & $\mathbf{0}$\\
$\mathbf{i}$ & $\mathbf{1}$ & $\mathbf{i}$ & $\mathbf{0}$\\
$\mathbf{0}$ & $\mathbf{1}$ & $\mathbf{1}$ & $\mathbf{1}$
\end{tabular}
}
\end{displaymath}
\noindent Note that the set of designated values, denoted by $\mathcal{D}_3$, is $\{ \mathbf{1}, \mathbf{i} \}$. The semantic consequence relation $\models_3$ is defined in terms of preservation of designated values.
\end{definition}

\begin{remark}
From a purely technical viewpoint, the above truth table for negation is exactly the one for the three-valued logic developed by {\L}ukasiewicz, as well as for the Logic of Paradox (cf. \cite{Priest79}). Moreover, the truth table for conditional is identical with the one in {\bf RM$_3^\supset$} (cf. \cite{Avron1986}), {\bf LFI1} (cf. \cite{CMdA2000}) and {\bf CLuNs} (cf. \cite{BatDeCl2004}), among many other systems.
\end{remark}

\begin{remark}
Note that in view of a general result established by Arnon Avron, Ofer Arieli and Anna Zamansky, it follows that $\models_3$ is maximally paraconsistent in the strong sense, and thus maximal with respect to extended classical logic, by \cite[Corollary 3.6]{AAZ2011}. 
\end{remark}

\subsection{Proof system for the disjunction-free fragment of {\bf D$_2$}}

I now turn to the proof theory which is presented in terms of a Hilbert-style calculus.
\begin{definition}\label{def:D2-}
The system {\bf D$_2^-$} consists of the following axiom schemata and a rule of inference, where $A\dequiv B$ abbreviates $(A\dto B)\rdland (B\dto A)$.

\noindent
\begin{minipage}{.6\textwidth}
\begin{align*}
& A \dto (B \dto A) \tag{Ax1} \\
& (A \dto (B \dto C))\dto ((A \dto B)  \dto (A \dto C)) \tag{Ax2} \\
& ((A\dto B)\dto A)\dto A \tag{Ax3} \\
& (A \rdland B) \dto A \tag{Ax4} \label{CKpqp}\\
& (A \rdland B) \dto B \tag{Ax5} \label{CKpqq}\\
& (C\dto A) \dto ((C\dto B) \dto (C\dto (A \rdland B))) \tag{Ax6}
\end{align*}
\end{minipage}
\begin{minipage}{.4\textwidth}
\begin{align*}
& (\Not A \dto A)\dto A \tag{Ax7} \label{CCNppp}\\
& \Not \Not A \dequiv A \tag{Ax8} \label{ENNpp}\\
& \Not (A \rdland B) \dequiv (B\dto \Not A) \tag{Ax9} \label{ENKpqCqNp} \\
& \Not (A\dto B) \dequiv (A \rdland \Not B) \tag{Ax10} \label{ENCpqKpNq}\\
& \frac{\ A \ \ \ \ \ A \dto B \ }{B}  \tag{MP}\label{MP}
\end{align*}
\end{minipage}

\medskip
\noindent Finally, $\Gamma \vdash A$ iff there is a sequence of formulas $B_1, \dots, B_n, A$ ($n\geq 0$), called \emph{a derivation}, such that every formula in the sequence either (i) belongs to $\Gamma$; (ii) is an axiom of {\bf D$_2^-$}; (iii) is obtained by (MP) from formulas preceding it in the sequence. 
\end{definition}

\begin{remark}
Note that the only unusual axiom in the literature of paraconsistent logic is \eqref{ENKpqCqNp}. 
\end{remark}

Before moving further, note that the deduction theorem holds for $\vdash$.

\begin{proposition}
For all $\Gamma\cup \{ A, B \} \subseteq \mathsf{Form}_r^-$, $\Gamma,A\vdash B$ iff $\Gamma\vdash A\dto B$.
\end{proposition}

\section{Soundness and Completeness for the three-valued semantics}

I now turn to prove that the proof system introduced in the previous section is sound and complete with respect to the three-valued semantics. 

\subsection{Soundness}
I begin with the soundness which is easy as usual.
\begin{proposition}[Soundness]\label{prop:soundness}
For all $\Gamma\cup \{ A \}\subseteq \mathsf{Form}_r^-$, if $\Gamma \vdash A$ then $\Gamma\models_3 A$.
\end{proposition}
\begin{proof}
By a straightforward verification that each instance of each axiom schema always takes a designated value, and that (MP) preserves designated values.
\end{proof}

\subsection{Completeness}

For the completeness, some terminologies are needed. To this end, I deploy those from \cite{Schumm1975henkin} with a slightly different term using \emph{non-trivial} instead of \emph{consistent}.

\begin{definition}[Schumm]
For $\Sigma \cup \{ B \}\subseteq \mathsf{Form}_r^-$, $\Sigma$ is \emph{maximally non-trivial} iff (i) $\Sigma\not\vdash A$ for some $A\in\mathsf{Form}_r^-$ and (ii) for every $A\in \mathsf{Form}_r^-$, if $A \not\in \Sigma$ then $\Sigma \cup \{ A \}\vdash B$ for all $B\in\mathsf{Form}_r^-$. 
\end{definition}


\begin{remark}
Note that if $\Sigma$ is maximally non-trivial, then $\Sigma$ is a theory, i.e. closed under $\vdash$. 
\end{remark}

Then the following well-known lemma is obtained. The proof is given in \cite[Theorem 8]{Schumm1975henkin}. 

\begin{lemma}[Schumm]\label{lem:Lindenbaum}
For all $\Sigma\cup \{ A \}\subseteq \mathsf{Form}_r^-$, suppose that $\Sigma\not\vdash A$. Then, there is a $\Pi \supseteq \Sigma$ such that $\Pi$ is maximally non-trivial and $A\not\in\Pi$.
\end{lemma}

Moreover, the following lemma, which will be useful later, is also easy to prove.

\begin{lemma}\label{lem:to-in-pdc}
If $\Sigma$ is maximally non-trivial, then $\Sigma\vdash A\dto B$ iff ($\Sigma\not\vdash A$ or $\Sigma\vdash B$).
\end{lemma}

\begin{definition}
Let $\Sigma$ be maximally non-trivial. Then, let $v_\Sigma$ from $\mathsf{Prop}$ to $\{\mathbf{1}, \mathbf{i}, \mathbf{0}\}$ be defined as follows:
\[ 
v_\Sigma (p){=}\mathbf{1} \text{ iff } \Sigma\not\vdash \Not p \text{\quad and \quad} 
v_\Sigma (p){=}\mathbf{i}  \text{ iff } \Sigma\vdash p \text{ and } \Sigma\vdash \Not p \text{\quad and \quad} 
v_\Sigma (p){=}\mathbf{0} \text{ iff } \Sigma\not \vdash p
\]
\end{definition}

I need one more lemma which is the key for the completeness result.
\begin{lemma}\label{lem:canonical-valuation} 
If $\Sigma$ is maximally non-trivial, then the following holds for all $B\in\mathsf{Form}_r^-$.
\[ 
v_\Sigma (B){=}\mathbf{1} \text{ iff } \Sigma\not\vdash \Not B \text{\quad and \quad} 
v_\Sigma (B){=}\mathbf{i}  \text{ iff } \Sigma\vdash B \text{ and } \Sigma\vdash \Not B \text{\quad and \quad} 
v_\Sigma (B){=}\mathbf{0} \text{ iff } \Sigma\not \vdash B
\]
\end{lemma}
\begin{proof}
Note first that the well-definedness of $v_\Sigma$ is obvious. Then the desired result is proved by induction on the the construction of $B$. The base case, for atomic formulas, is obvious by the definition. For the induction step, the cases are split based on the connectives. 

\smallskip

\noindent {\bf Case 1.} If $B=\Not C$, then there are the following three cases.
\begin{align*}
v_\Sigma(\Not C)=\mathbf{1} 
                         &\text{ iff } v_\Sigma(C)=\mathbf{0} & \text{by the definition of $v_\Sigma$} \\
                         &\text{ iff } \Sigma\not\vdash C & \text{by IH} \\
                         &\text{ iff } \Sigma\not\vdash \Not \Not C & \text{by \eqref{ENNpp}} 
\end{align*}
\begin{align*}
v_\Sigma(\Not C)=\mathbf{i}
                         &\text{ iff } v_\Sigma(C)=\mathbf{i} & \text{by the definition of $v_\Sigma$} \\
                         &\text{ iff } \Sigma\vdash \Not C \text{ and } \Sigma\vdash C & \text{by IH}  \\
                         &\text{ iff } \Sigma\vdash \Not C \text{ and } \Sigma\vdash \Not \Not C & \text{by \eqref{ENNpp}}
\end{align*}
\begin{align*}
v_\Sigma(\Not C)=\mathbf{0}
                         &\text{ iff } v_\Sigma(C)=\mathbf{1} & \text{by the definition of $v_\Sigma$} \\
                         &\text{ iff } \Sigma\not\vdash \Not C & \text{by IH}
\end{align*}

\noindent {\bf Case 2.} If $B=C {\dto} D$, then there are the following three cases.
\begin{align*}
v_\Sigma(C\dto D)=\mathbf{1}
                         &\text{ iff } v_\Sigma(C)=\mathbf{0} \text{ or } v_\Sigma(D)=\mathbf{1} & \text{by the definition of $v_\Sigma$} \\
                         &\text{ iff } \Sigma\not\vdash C \text{ or } \Sigma\not\vdash \Not D & \text{by IH} \\
                         &\text{ iff } \Sigma\not\vdash (C\land \Not D) & \text{by $\Sigma$ is a theory} \\
                         &\text{ iff } \Sigma\not\vdash \Not (C\dto D) & \text{by \eqref{ENCpqKpNq}}
\end{align*}
\begin{align*}
v_\Sigma(C\dto D)=\mathbf{i}
                         &\text{ iff } v_\Sigma(C)\neq\mathbf{0} \text{ and } v_\Sigma(D)=\mathbf{i} & \text{by the definition of $v_\Sigma$} \\
                         &\text{ iff } \Sigma\vdash C \text{ and } (\Sigma\vdash D \text{ and } \Sigma\vdash \Not D) & \text{by IH} \\
                         &\text{ iff } (\Sigma\not\vdash C \text{ or } \Sigma\vdash D)\text{ and } \Sigma\vdash (C\land \Not D) & \text{$\Sigma$ is a theory} \\
                         &\text{ iff } \Sigma\vdash (C\dto D) \text{ and }\Sigma\vdash \Not (C\dto D) & \text{by Lemma \ref{lem:to-in-pdc} and \eqref{ENCpqKpNq}}
\end{align*}
\begin{align*}
v_\Sigma(C\dto D)=\mathbf{0}
                         &\text{ iff } v_\Sigma(C)\neq\mathbf{0} \text{ and } v_\Sigma(D)=\mathbf{0} & \text{by the definition of $v_\Sigma$} \\
                         &\text{ iff } \Sigma\vdash C \text{ and } \Sigma\not\vdash D & \text{by IH} \\
                         &\text{ iff } \Sigma\not\vdash (C\dto D) & \text{by Lemma \ref{lem:to-in-pdc}}
\end{align*}

\noindent {\bf Case 3.} If $B=C {\land} D$, then there are the following three cases.
\begin{align*}
v_\Sigma(C\land D)=\mathbf{1}
                         &\text{ iff } v_\Sigma(C)=\mathbf{1} \text{ and } v_\Sigma(D)\neq\mathbf{0} & \text{by the definition of $v_\Sigma$} \\
                         &\text{ iff } \Sigma\vdash D \text{ and } \Sigma\not\vdash \Not C & \text{by IH } \\
                         &\text{ iff } \Sigma\not\vdash D\dto \Not C & \text{by Lemma \ref{lem:to-in-pdc}} \\
                         &\text{ iff } \Sigma\not\vdash \Not (C\land D) & \text{by \eqref{ENKpqCqNp}}
\end{align*}
\begin{align*}
v_\Sigma(C{\land} D)=\mathbf{i}
                         &\text{ iff } v_\Sigma(C)=\mathbf{i} \text{ and } v_\Sigma(D)\neq\mathbf{0} & \text{by the definition of $v_\Sigma$} \\
                         &\text{ iff } (\Sigma\vdash C \text{ and } \Sigma\vdash \Not C) \text{ and } \Sigma\vdash D & \text{by IH} \\
                         &\text{ iff } (\Sigma\vdash C \text{ and } \Sigma\vdash D)\text{ and } (\Sigma\not\vdash D \text{ or } \Sigma\vdash \Not C)& \text{by simple calculation} \\
                         &\text{ iff } (\Sigma\vdash C \text{ and } \Sigma\vdash D)\text{ and } \Sigma\vdash D\dto \Not C& \text{by Lemma \ref{lem:to-in-pdc}} \\
                         &\text{ iff } \Sigma\vdash (C\land D) \text{ and }\Sigma\vdash \Not (C\land D) & \text{$\Sigma$ is a theory and by \eqref{ENCpqKpNq}} 
\end{align*}
\begin{align*}
v_\Sigma(C\land D)=\mathbf{0}
                         &\text{ iff } v_\Sigma(C)=\mathbf{0} \text{ or } v_\Sigma(D)=\mathbf{0} & \text{by the definition of $v_\Sigma$} \\
                         &\text{ iff } \Sigma\not\vdash C \text{ or } \Sigma\not\vdash D & \text{by IH} \\
                         &\text{ iff } \Sigma\not\vdash (C\land D) & \text{$\Sigma$ is a theory} 
\end{align*}

This completes the proof.
\end{proof}

\begin{theorem}[Completeness] \label{thm:completenesswrt3}
For all $\Gamma\cup \{ A \}\subseteq \mathsf{Form}_r^-$, if $\Gamma\models_3 A$ then $\Gamma \vdash A$.
\end{theorem}
\begin{proof}
Assume $\Gamma\not\vdash A$. Then, by Lemma~\ref{lem:Lindenbaum}, there is a $\Pi\supseteq \Gamma$ such that $\Pi$ is maximally non-trivial and $A\not\in\Pi$, and by Lemma~\ref{lem:canonical-valuation}, a three-valued {\bf D$_2^-$}-valuation $v_{\Pi}$ can be defined with $I_{\Pi}(B)\in\mathcal{D}_3$ for every $B\in\Gamma$ and $I_{\Pi}(A)\not\in\mathcal{D}_3$. Thus it follows that $\Gamma \not\models_3 A$, as desired.
\end{proof}

\section{The main result}
By making use of the result in the previous section, I prove the main result of this article. To this end,  I need one more lemma. 
\begin{lemma}\label{lem:4implies3}
For all $\Gamma\cup \{ A \}\subseteq \mathsf{Form}_r^-$, if $\Gamma\models_4 A$ then $\Gamma\models_3 A$.
\end{lemma}
\begin{proof}
Suppose $\Gamma\not\models_3 A$. Then there is a three-valued {\bf D$_2^-$}-interpretation $v_0$ such that $I_0(B)\in \mathcal{D}_3$ for all $B\in \Gamma$ and $I_0(A)\not\in\mathcal{D}_3$. Now, let $v_1$ be a four-valued {\bf D$_2^-$}-interpretation such that $v_1(p)=v_0(p)$. Then, it holds that $I_1(A)=\mathbf{1}$ iff $I_0(A)=\mathbf{1}$ and $I_1(A)=\mathbf{0}$ iff $I_0(A)=\mathbf{0}$. This can be proved by a simple induction on the complexity of $A$. 
\begin{itemize}
\setlength{\parskip}{0cm}
\setlength{\itemsep}{0cm}
\item The base case when $A\in \mathsf{Prop}$ is obvious by definition.
\item For induction step, consider the following two cases.

\begin{itemize}
\setlength{\parskip}{0cm}
\setlength{\itemsep}{0cm}
\item If $A$ is of the form $\Not B$, then by IH,  

 \begin{itemize}
  \setlength{\parskip}{0cm}
  \setlength{\itemsep}{0cm}
\item $I_1 (B)=\mathbf{1}$ iff $I_0(B)=\mathbf{1}$ and 
  \item $I_1 (B)=\mathbf{0}$ iff $I_0(B)=\mathbf{0}$.
 \end{itemize} 
Then, by the truth table, it follows that $I_1 (\Not B){=}\mathbf{0}$ iff (by the truth table) $I_1(B){=}\mathbf{1}$ iff (by IH)  $I_0(B){=}\mathbf{1}$ iff (by the truth table) $I_0(\Not B){=}\mathbf{0}$. Moreover, $I_1 (\Not B){=}\mathbf{1}$ iff (by the truth table) $I_1(B){=}\mathbf{0}$ iff (by IH) $I_0(B){=}\mathbf{0}$ iff (by the truth table) $I_0(\Not B){=}\mathbf{1}$.
\item If $A$ is of the form $B\dto C$, then by IH,
 
 \begin{itemize}
\setlength{\parskip}{0cm}
\setlength{\itemsep}{0cm}
  \item $I_1 (B)=\mathbf{1}$ iff $I_0(B)=\mathbf{1}$, $I_1 (B)=\mathbf{0}$ iff $I_0(B)=\mathbf{0}$, and
  \item $I_1 (C)=\mathbf{1}$ iff $I_0(C)=\mathbf{1}$, $I_1 (C)=\mathbf{0}$ iff $I_0(C)=\mathbf{0}$. 
 \end{itemize}
Then, by the truth table, it follows that $I_1 (B\dto C){=}\mathbf{0}$ iff (by the truth table) $I_1(B){\neq}\mathbf{0}$ and $I_1(C){=}\mathbf{0}$ iff (by IH) $I_0(B)\neq\mathbf{0}$ and $I_0(C){=}\mathbf{0}$ iff (by the truth table) $I_0(B\dto C){=}\mathbf{0}$. Moreover, $I_1 (B\dto C){=}\mathbf{1}$ iff (by the truth table) $I_1(B){=}\mathbf{0}$ or $I_1(C){=}\mathbf{1}$ iff (by IH) $I_0(B){=}\mathbf{0}$ and $I_0(C){=}\mathbf{1}$ iff (by the truth table) $I_0(B\dto C){=}\mathbf{1}$.
\end{itemize}
The case for conjunction is similar to the case for $\to_d$. This completes the proof. 
\end{itemize} 
Once this is established it is easy to see that the desired result holds since $I_1(A)=\mathbf{0}$ iff $I_0(A)=\mathbf{0}$ is equivalent to $I_1(A)\not\in\mathcal{D}_4$ iff $I_0(A)\not\in\mathcal{D}_3$.
\end{proof}

I am now ready to prove the main result. 
 
\begin{theorem}[Main Theorem]\label{thm:3iffK}
For all $\Gamma\cup \{ A \}\subseteq \mathsf{Form}_r^-$, $\Gamma \models_3 A$ iff $\Gamma\models_d A$.
\end{theorem}
\begin{proof}
For the left-to-right direction, if $\Gamma\models_3 A$ then $\Gamma\vdash A$ by Theorem~\ref{thm:completenesswrt3}. One may then check that if $\Gamma\vdash A$ then $\Gamma\models_d A$. This is tedious but not difficult. 
For the other direction, if $\Gamma\models_d A$ then it immediately implies that $\Gamma\models_4 A$, by recalling Remark~\ref{rem:special-case}. Thus, together with Lemma~\ref{lem:4implies3}, the desired result is proved.
\end{proof}

As a corollary of Proposition~\ref{prop:soundness} and Theorems~\ref{thm:completenesswrt3} and \ref{thm:3iffK}, the following result is obtained.
\begin{corollary}
For all $\Gamma\cup \{ A \}\subseteq \mathsf{Form}_r^-$, $\Gamma\vdash A$ iff $\Gamma\models_d A$.
\end{corollary}

\section{Reflections}

\subsection{The language $\mathcal{L}_l$}\label{subsec:left-discussive-conjunction}
In the later works related to discussive logics, the language $\mathcal{L}_l$ has been also studied intensively. Here, I note that the above observations carry over to $\mathcal{L}_l^-$.

\begin{itemize}
\setlength{\parskip}{0cm}
\setlength{\itemsep}{0cm}
\item First, the truth condition for the left discussive conjunction within the Kripke semantics is as follows.

\begin{center}
\begin{tabular}{rl}
($\land_d^l$)& $v(w, A \land_d^l B)=1$  iff  $ \text{for some } x\in W (v(x, A)=1) \text{ and } v(w, B)=1$.
\end{tabular}
\end{center}
\item Second, the three- and four-valued truth tables for the left discussive conjunction are as follows. Of course, the four-valued truth table is obtained by considering the special case of the Kripke semantics (recall Remark~\ref{rem:special-case}), and the three-valued truth table is obtained by ``merging'' the intermediate values in the four-valued truth table.
\begin{displaymath}
{\small
\begin{tabular}{c|cccc}
$A \land_d^l B$ & $\mathbf{1}$ & $\mathbf{i}$ &  $\mathbf{0}$ \\ \hline
$\mathbf{1}$ & $\mathbf{1}$ & $\mathbf{i}$ & $\mathbf{0}$\\
$\mathbf{i}$ & $\mathbf{1}$ & $\mathbf{i}$ & $\mathbf{0}$\\
$\mathbf{0}$ & $\mathbf{0}$ & $\mathbf{0}$ &  $\mathbf{0}$\\
\end{tabular}
\quad
\begin{tabular}{c|cccc}
$A \land_d^l B$ & $\mathbf{1}$ & $\mathbf{i}$ & $\mathbf{j}$ & $\mathbf{0}$ \\ \hline
$\mathbf{1}$ & $\mathbf{1}$ & $\mathbf{i}$ & $\mathbf{j}$ & $\mathbf{0}$\\
$\mathbf{i}$ & $\mathbf{1}$ & $\mathbf{i}$ & $\mathbf{j}$ & $\mathbf{0}$\\
$\mathbf{j}$ & $\mathbf{1}$ & $\mathbf{i}$ & $\mathbf{j}$ & $\mathbf{0}$\\
$\mathbf{0}$ & $\mathbf{0}$ & $\mathbf{0}$ & $\mathbf{0}$ & $\mathbf{0}$\\
\end{tabular}
}
\end{displaymath}
\item Third, for the proof system, \eqref{ENKpqCqNp} is replaced by the following.
\[ \Not (A \land_d^l B) \dequiv (A\dto \Not B) \tag{Ax9'} \label{ENKpqCpNq} \] 
\end{itemize}
Based on these, the equivalence of the discussive semantics and the three-valued semantics may be established in a similar manner. 
For those who are interested in the details, note that for the purpose of establishing the result corresponding to Theorem~\ref{thm:3iffK}, it suffices to check the following three items.
\begin{itemize}
\setlength{\parskip}{0cm}
\setlength{\itemsep}{0cm}
\item Lemma~\ref{lem:canonical-valuation}, for the completeness result, i.e. if $\Gamma\models_3 A$ then $\Gamma\vdash A$.
\item $\Gamma\vdash A$ then $\Gamma\models_d A$.
\item Lemma~\ref{lem:4implies3}, i.e. $\Gamma\models_4 A$ then $\Gamma\models_3 A$.
\end{itemize}
For the first item, it is enough to check the case related to conjunction, in particular the following two cases.
\begin{align*}
v_\Sigma(C\land D)=\mathbf{1}
                         &\text{ iff } v_\Sigma(C)\neq\mathbf{0} \text{ and } v_\Sigma(D)=\mathbf{1} & \text{by the definition of $v_\Sigma$} \\
                         &\text{ iff } \Sigma\vdash C \text{ and } \Sigma\not\vdash \Not D & \text{by IH } \\
                         &\text{ iff } \Sigma\not\vdash C\dto \Not D & \text{by Lemma \ref{lem:to-in-pdc}} \\
                         &\text{ iff } \Sigma\not\vdash \Not (C\land D) & \text{by \eqref{ENKpqCpNq}}
\end{align*}
\begin{align*}
v_\Sigma(C{\land} D)=\mathbf{i}
                         &\text{ iff } v_\Sigma(C)\neq\mathbf{0} \text{ and } v_\Sigma(D)=\mathbf{i} & \text{by the definition of $v_\Sigma$} \\
                         &\text{ iff } \Sigma\vdash C  \text{ and } (\Sigma\vdash D \text{ and } \Sigma\vdash \Not D) & \text{by IH} \\
                         &\text{ iff } (\Sigma\vdash C \text{ and } \Sigma\vdash D)\text{ and } (\Sigma\not\vdash C \text{ or } \Sigma\vdash \Not D)& \text{by simple calculation} \\
                         &\text{ iff } (\Sigma\vdash C \text{ and } \Sigma\vdash D)\text{ and } \Sigma\vdash C\dto \Not D& \text{by Lemma \ref{lem:to-in-pdc}} \\
                         &\text{ iff } \Sigma\vdash (C\land D) \text{ and }\Sigma\vdash \Not (C\land D) & \text{$\Sigma$ is a theory and by \eqref{ENKpqCpNq}} 
\end{align*}

For the second item, this is immediate in view of the new truth condition for the left discussive conjunction within the Kripke semantics. 

Finally, for the third item, it is again enough to check the case for conjunction, and the proof runs as follows. If $A$ is of the form $B\land_d^l C$, then by IH,
 
\begin{itemize}
\setlength{\parskip}{0cm}
\setlength{\itemsep}{0cm}
\item $I_1 (B)=\mathbf{1}$ iff $I_0(B)=\mathbf{1}$, $I_1 (B)=\mathbf{0}$ iff $I_0(B)=\mathbf{0}$, and
\item $I_1 (C)=\mathbf{1}$ iff $I_0(C)=\mathbf{1}$, $I_1 (C)=\mathbf{0}$ iff $I_0(C)=\mathbf{0}$. 
\end{itemize}
Then, by the truth table, it follows that $I_1 (B\land_d^l C){=}\mathbf{0}$ iff (by the truth table) $I_1(B){=}\mathbf{0}$ or $I_1(C){=}\mathbf{0}$ iff (by IH) $I_0(B){=}\mathbf{0}$ and $I_0(C){=}\mathbf{0}$ iff (by the truth table) $I_0(B\land_d^l C){=}\mathbf{0}$. Moreover, $I_1 (B\land_d^l C){=}\mathbf{1}$ iff (by the truth table) $I_1(B){\neq}\mathbf{0}$ or $I_1(C){=}\mathbf{1}$ iff (by IH) $I_0(B){\neq}\mathbf{0}$ and $I_0(C){=}\mathbf{1}$ iff (by the truth table) $I_0(B\land_d^l C){=}\mathbf{1}$.

Based on these, the proof of Theorem~\ref{thm:3iffK} can be repeated to establish the desired result.

\subsection{The language $\mathcal{L}$}\label{subsec:non-discussive-conjunction}
If one considers the very first discussive language $\mathcal{L}$ in which the only discussive connective is conditional, a similar result is obtained by considering the negation-conditional fragment. More specifically, the concerned fragment is equivalent to the three-valued semantics induced by the following truth tables:
\begin{displaymath}
{\small
\begin{tabular}{c|c}
$A$  &  $\Not A$\\
\hline
$\mathbf{1}$ &  $\mathbf{0}$\\
$\mathbf{i}$ & $\mathbf{i}$\\
$\mathbf{0}$  &  $\mathbf{1}$\\
\end{tabular}
\quad
\begin{tabular}{c|cccc}
$A \dto B$ & $\mathbf{1}$ & $\mathbf{i}$  & $\mathbf{0}$ \\ \hline
$\mathbf{1}$ & $\mathbf{1}$ & $\mathbf{i}$ & $\mathbf{0}$\\
$\mathbf{i}$ & $\mathbf{1}$ & $\mathbf{i}$ & $\mathbf{0}$\\
$\mathbf{0}$ & $\mathbf{1}$ & $\mathbf{1}$ & $\mathbf{1}$
\end{tabular}
}
\end{displaymath}
This can be confirmed by carefully removing the cases for conjunction in the proof of the main result. For those who are interested in the details, note once again that for the purpose of establishing the result corresponding to Theorem~\ref{thm:3iffK}, it suffices to check the following three items.
\begin{itemize}
\setlength{\parskip}{0cm}
\setlength{\itemsep}{0cm}
\item Lemma~\ref{lem:canonical-valuation}, for the completeness result, i.e. if $\Gamma\models_3 A$ then $\Gamma\vdash A$.
\item $\Gamma\vdash A$ then $\Gamma\models_d A$.
\item Lemma~\ref{lem:4implies3}, i.e. $\Gamma\models_4 A$ then $\Gamma\models_3 A$.
\end{itemize}
In particular, it is enough to check that the previous proofs are not essentially relying on conjunction. For the first item, note first that \eqref{ENCpqKpNq} needs to be replaced by the following three axioms.
\begin{align*}
& \Not (A\dto B) \dto A \tag{Ax10.1} \label{ENCpqKpNq1}\\
& \Not (A\dto B) \dto \Not B \tag{Ax10.2} \label{ENCpqKpNq2}\\
& A\dto (\Not B\dto \Not (A\dto B)) \tag{Ax10.3} \label{ENCpqKpNq3}
\end{align*}
Then, it suffices to check that if $\Sigma$ is maximally non-trivial, then $\Sigma\vdash \Not (A\dto B)$ iff ($\Sigma\vdash A$ and $\Sigma\vdash \Not B$). This of course holds even without the maximal non-triviality. For the second and the third items, there is nothing to be checked since they are both already established.

Based on these, the proof of Theorem~\ref{thm:3iffK} can be repeated to establish the desired result.
 
\subsection{Discussive negation}\label{subsec:discussive-negation}
Another variation of the main result is obtained by considering a discussive interpretation of negation, a suggestion made by Jerzy Perzanowski as one of the comments of the translator in \cite[p.59]{J2}, and explored further by Ciuciura in \cite{C1}.\footnote{There is, unfortunately, a problem with one of the main results in \cite{C1}. See the appendix of \cite{OAD2} for the details.} Here, once again, I note that the above observations carry over to this variant. 

\begin{itemize}
\setlength{\parskip}{0cm}
\setlength{\itemsep}{0cm}
\item First, the truth condition for the discussive negation within the Kripke semantics is as follows.

\begin{center}
\begin{tabular}{rl}
($\Not_d$)& $v(w, \Not_d A)=1$  iff $ \text{for some } x\in W, v(x, A)=0$.
\end{tabular}
\end{center}
\item Second, the three- and four-valued truth tables for the discussive negation are as follows.
\begin{displaymath}
{\small
\begin{tabular}{c|c}
$A$  &  $\Not_d A$\\ \hline
$\mathbf{1}$ &  $\mathbf{0}$\\
$\mathbf{i}$ & $\mathbf{1}$\\
$\mathbf{0}$ &  $\mathbf{1}$\\
\end{tabular}
\quad
\begin{tabular}{c|c}
$A$  &  $\Not_d A$\\ \hline
$\mathbf{1}$ &  $\mathbf{0}$\\
$\mathbf{i}$ & $\mathbf{1}$\\
$\mathbf{j}$ & $\mathbf{1}$\\
$\mathbf{0}$ &  $\mathbf{1}$\\
\end{tabular}
}
\end{displaymath}
\item Third, for the proof system, \eqref{ENNpp} is replaced by the following.
\[ \Not_d A \dto (\Not_d\Not_d A\dto  B) \] 
\end{itemize}
Based on these, the equivalence of the discussive semantics and the three-valued semantics is established in a similar manner. I first note here that given the proof system, the following is obtained. 
\begin{lemma}\label{lem:4Ciuciura}
If $\Sigma$ is maximally non-trivial, then $\Sigma\vdash \Not_d\Not_d A$ iff $\Sigma\not\vdash \Not_d A$.
\end{lemma}
Then, for the purpose of establishing the result corresponding to Theorem~\ref{thm:3iffK}, it suffices to check the following three items.
\begin{itemize}
\setlength{\parskip}{0cm}
\setlength{\itemsep}{0cm}
\item Lemma~\ref{lem:canonical-valuation}, for the completeness result, i.e. if $\Gamma\models_3 A$ then $\Gamma\vdash A$.
\item $\Gamma\vdash A$ then $\Gamma\models_d A$.
\item Lemma~\ref{lem:4implies3}, i.e. $\Gamma\models_4 A$ then $\Gamma\models_3 A$.
\end{itemize}
For the first item, it is enough to check the case related to negation, in particular the following case since negated formula \emph{never} takes the value $\mathbf{i}$, and the case when negated formula takes the value $\mathbf{0}$ is already covered by the original Lemma~\ref{lem:canonical-valuation}.
\begin{align*}
v_\Sigma(\Not C)=\mathbf{1} 
                         &\text{ iff } v_\Sigma(C)\neq\mathbf{1} & \text{by the definition of $v_\Sigma$} \\
                         &\text{ iff } \Sigma\vdash \Not C & \text{by IH} \\
                         &\text{ iff } \Sigma\not\vdash \Not \Not C & \text{by Lemma~\ref{lem:4Ciuciura}} 
\end{align*}

For the second item, this is immediate in view of the new truth condition for the discussive negation within the Kripke semantics. 

Finally, for the third item, it is sufficient to check the case for negation, and the proof runs as follows. If $A$ is of the form $\Not B$, then by IH,  
 \begin{itemize}
 \setlength{\parskip}{0cm}
\setlength{\itemsep}{0cm}
 \item $I_1 (B)=\mathbf{1}$ iff $I_0(B)=\mathbf{1}$ and 
  \item $I_1 (B)=\mathbf{0}$ iff $I_0(B)=\mathbf{0}$.
 \end{itemize} 
\smallskip
Then, by the truth table, it follows that $I_1 (\Not B){=}\mathbf{0}$ iff (by the truth table) $I_1(B){=}\mathbf{1}$ iff (by IH)  $I_0(B){=}\mathbf{1}$ iff (by the truth table) $I_0(\Not B){=}\mathbf{0}$. Moreover, $I_1 (\Not B){=}\mathbf{1}$ iff (by the truth table) $I_1(B){\neq}\mathbf{1}$ iff (by IH) $I_0(B){\neq}\mathbf{1}$ iff (by the truth table) $I_0(\Not B){=}\mathbf{1}$.

Based on these, the proof of Theorem~\ref{thm:3iffK} can be repeated to establish the desired result.

\subsection{Disjunction}\label{subsec:non-discussive-disjunction}
One may wonder about the possibility of adding disjunction to the many-valued semantics. In the case of three-valued semantics, one can prove the completeness in a similar manner. 
\begin{itemize}
\setlength{\parskip}{0cm}
\setlength{\itemsep}{0cm}
\item First, let {\bf D$_2^+$} be the expansion of {\bf D$_2^-$} obtained by adding the following axiom schemata.

\noindent
\begin{minipage}{.3\textwidth}
\begin{align*}
& A \dto (A \lor B) \tag{Ax13} \\
& B \dto (A \lor B) \tag{Ax14}
\end{align*}
\end{minipage}
\begin{minipage}{.6\textwidth}
\begin{align*}
& (A \dto C) \dto ((B \dto C) \dto ((A \lor B) \dto C)) \tag{Ax15} \label{Ax15} \\
& \Not (A \lor B) \dequiv (\Not A \rdland \Not B) \tag{Ax16} \label{Ax16}
\end{align*}
\end{minipage}

\smallskip

\noindent The consequence relation $\vdash_\mathbf{D_2^+}$ is defined as before.

\item Second, the three-valued truth tables for \emph{{\bf D$_2^+$}-valuation} are as follows:
\begin{displaymath}
{\small
\begin{tabular}{c|c}
$A$  & $\Not A$ \\
\hline
$\mathbf{1}$ &  $\mathbf{0}$ \\
$\mathbf{i}$ & $\mathbf{i}$\\
$\mathbf{0}$ &  $\mathbf{1}$ \\
\end{tabular}
\quad
\begin{tabular}{c|cccc}
$A \lor B$ & $\mathbf{1}$ & $\mathbf{i}$  & $\mathbf{0}$ \\ \hline
$\mathbf{1}$ & $\mathbf{1}$ & $\mathbf{1}$ & $\mathbf{1}$\\
$\mathbf{i}$ & $\mathbf{1}$ & $\mathbf{i}$ & $\mathbf{i}$\\
$\mathbf{0}$ & $\mathbf{1}$ & $\mathbf{i}$ & $\mathbf{0}$
\end{tabular}
\quad
\begin{tabular}{c|cccc}
$A \rdland B$ & $\mathbf{1}$ & $\mathbf{i}$ &  $\mathbf{0}$ \\ \hline
$\mathbf{1}$ & $\mathbf{1}$ & $\mathbf{1}$ & $\mathbf{0}$\\
$\mathbf{i}$ & $\mathbf{i}$ & $\mathbf{i}$ & $\mathbf{0}$\\
$\mathbf{0}$ & $\mathbf{0}$ & $\mathbf{0}$ &  $\mathbf{0}$\\
\end{tabular}
\quad
\begin{tabular}{c|cccc}
$A \dto B$ & $\mathbf{1}$ & $\mathbf{i}$  & $\mathbf{0}$ \\ \hline
$\mathbf{1}$ & $\mathbf{1}$ & $\mathbf{i}$ & $\mathbf{0}$\\
$\mathbf{i}$ & $\mathbf{1}$ & $\mathbf{i}$ & $\mathbf{0}$\\
$\mathbf{0}$ & $\mathbf{1}$ & $\mathbf{1}$ & $\mathbf{1}$
\end{tabular}
}
\end{displaymath}
\noindent The designated values are $\mathbf{1}$ and $\mathbf{i}$, and the semantic consequence relation $\models_3^+$ is defined in terms of preservation of designated values.
\end{itemize}
Then, the main result will carry over to this expansion of {\bf D$_2^-$}. This time, I first note here that given the proof system, the following is obtained. 
\begin{lemma}\label{lem:4addingdisjunction}
If $\Sigma$ is maximally non-trivial, then $\Sigma\vdash_\mathbf{D_2^+} A{\lor} B$ iff $\Sigma\vdash_\mathbf{D_2^+} A$ or $\Sigma\vdash_\mathbf{D_2^+} B$.
\end{lemma}
Then, for the purpose of establishing the soundness and completeness results, the soundness is straightforward. For the completeness result, it suffices to check the additional case for Lemma~\ref{lem:canonical-valuation} related to disjunction since other cases are already covered. If $B=C {\lor} D$, then there are the following three cases.
\begin{align*}
v_\Sigma(C{\lor} D)=\mathbf{1}
                         &\text{ iff } v_\Sigma(C)=\mathbf{1} \text{ or } v_\Sigma(D)=\mathbf{1} & \text{by the definition of $v_\Sigma$} \\
                         &\text{ iff } \Sigma\not\vdash \Not C \text{ or } \Sigma\not\vdash \Not D & \text{by IH} \\
                         &\text{ iff } \Sigma\not\vdash (\Not C\land \Not D) & \text{by $\Sigma$ is a theory} \\
                         &\text{ iff } \Sigma\not\vdash \Not (C{\lor} D) & \text{by \eqref{Ax16}}
\end{align*}
\begin{align*}
v_\Sigma(C{\lor} D){=}\mathbf{i}
                         &\text{ iff } (v_\Sigma(C){\neq}\mathbf{1} \text{ and } v_\Sigma(D){=}\mathbf{i}) \text{ or } \\
                         &\qquad \qquad \qquad  (v_\Sigma(D){\neq}\mathbf{1} \text{ and } v_\Sigma(C){=}\mathbf{i}) & \text{by the def. of $v_\Sigma$} \\
                         &\text{ iff } (\Sigma\vdash \Not C \text{ and } (\Sigma\vdash D \text{ and } \Sigma\vdash \Not D)) \text{ or } \\
                         &\qquad \qquad \qquad  (\Sigma\vdash \Not D \text{ and } (\Sigma\vdash C \text{ and } \Sigma\vdash \Not C))& \text{by IH} \\
                         &\text{ iff } (\Sigma\vdash C \text{ or } \Sigma\vdash D)\text{ and } \Sigma\vdash (\Not C\land \Not D) & \text{$\Sigma$ is a theory} \\
                         &\text{ iff } \Sigma\vdash C{\lor} D \text{ and }\Sigma\vdash \Not (C{\lor} D) & \text{by Lemma \ref{lem:4addingdisjunction} and \eqref{Ax16}}
\end{align*}
\begin{align*}
v_\Sigma(C{\lor} D)=\mathbf{0}
                         &\text{ iff } v_\Sigma(C)=\mathbf{0} \text{ and } v_\Sigma(D)=\mathbf{0} & \text{by the definition of $v_\Sigma$} \\
                         &\text{ iff } \Sigma\not\vdash C \text{ and } \Sigma\not\vdash D & \text{by IH} \\
                         &\text{ iff } \Sigma\not\vdash (C{\lor} D) & \text{by Lemma \ref{lem:4addingdisjunction}}
\end{align*}
Based on these, the desired result is obtained.

Note finally that neither {\bf D$_2^+$} nor {\bf D$_2$} contains the other. Indeed, the following may be verified. 
\begin{itemize}
\setlength{\parskip}{0cm}
\setlength{\itemsep}{0cm}
\item $\vdash_\mathbf{D_2} \Not (A\lor \Not A)\to_d B$ but $\not\vdash_\mathbf{D_2^+} \Not (A\lor \Not A)\to_d B$,
\item $\vdash_\mathbf{D_2^+} \Not (A \lor B) \dequiv (\Not A \rdland \Not B)$ but $\not\vdash_\mathbf{D_2} \Not (A \lor B) \dequiv (\Not A \rdland \Not B)$.
\end{itemize}

\subsection{An application}
The main result was obtained rather surprisingly by looking at the semantics for discussive logics without any aim of bridging discussive logics and many-valued logics. However, in view of the relation between discussive semantics and many-valued semantics, one may change the perspective to regard discussive semantics as a tool to make sense of some of the many-valued logics. What I have in mind here are the semantic frameworks such as Michael Dunn's relational semantics (cf.~\cite{Dunn1976}), Richard and Valerie Routley's star semantics (cf.~\cite{Routleys1972}), and Graham Priest's plurivalent semantics (cf.~\cite{Priest1984hyper, Priest2014}). These can be seen as offering alternative two-valued semantics for many-valued logics, and by doing so these frameworks offer different ways to give intuitive readings to the additional truth values, and understand the semantics for the connectives. Indeed, the first two frameworks offer alternative semantics for the four-valued logic {\bf FDE}, and the last framework offers alternative semantics for {\bf LP} and weak Kleene logic, among many others.\footnote{For some of the recent discussions on this theme, see \cite{OmoriISMVL22,Omori2023-theoria,Omori2023finetti} which build heavily on \cite{herzberger1973dimensions,Haack1978}.}

In fact, the idea is already applied successfully to {\bf P$^1$} of Antonio Sette which is one of the oldest three-valued paraconsistent logics introduced in \cite{Sette1973}. More specifically, with the help of discussive semantics, one may intuitively read the three values with some discussive flavor, and moreover understand the paraconsistent negation as a negative modality. Further details, including a comparison to the so-called society semantics for {\bf P$^1$} devised by Walter Carnielli and Mamede Lima-Marques in \cite{Carnielli1999society}, can be found in \cite{OmoriLORI2017}. 

What I would like to add here is one more instance that seems to offer an alternative perspective to a variant of {\bf FDE}, called {\bf NFL} in \cite{Shramko2017}, and compare with {\bf FDE} as well as {\bf ETL}, introduced in \cite{Pietz2013nothing} (see also \cite{Marcos2011}). The rest of this subsection is devoted to spell out the details. Note that the language of {\bf FDE}, which consists of a finite set $\{ \Not, \land, \lor \}$ of propositional connectives and a countable set $\mathsf{Prop}$ of propositional variables, is referred to as $\mathcal{L}_{\bf FDE}$. Moreover, as expected, the set of formulas defined as usual in $\mathcal{L}_{\bf FDE}$ is denoted by $\mathsf{Form}_{\bf FDE}$.

\begin{definition}\label{def:four-valuedFDE}
A \emph{four-valued Belnap-Dunn-valuation} for $\mathcal{L}_{\bf FDE}$ is a homomorphism from $\mathsf{Form}_{\bf FDE}$ to $\{\mathbf{t}, \mathbf{b}, \mathbf{n}, \mathbf{f}\}$, induced by the following matrices:
\begin{displaymath}
{\small
\begin{tabular}{c|c}
$A$  &  $\Not A$\\ \hline
$\mathbf{t}$ &  $\mathbf{f}$\\
$\mathbf{b}$ & $\mathbf{b}$\\
$\mathbf{n}$ & $\mathbf{n}$\\
$\mathbf{f}$ &  $\mathbf{t}$\\
\end{tabular}
\quad
\begin{tabular}{c|cccc}
$A \lor B$ & $\mathbf{t}$ & $\mathbf{b}$ & $\mathbf{n}$ & $\mathbf{f}$ \\ \hline
$\mathbf{t}$ & $\mathbf{t}$ & $\mathbf{t}$ & $\mathbf{t}$ & $\mathbf{t}$\\
$\mathbf{b}$ & $\mathbf{t}$ & $\mathbf{b}$ & $\mathbf{t}$ & $\mathbf{b}$\\
$\mathbf{n}$ & $\mathbf{t}$ & $\mathbf{t}$ & $\mathbf{n}$ & $\mathbf{n}$\\
$\mathbf{f}$ & $\mathbf{t}$ & $\mathbf{b}$ & $\mathbf{n}$ & $\mathbf{f}$
\end{tabular}
\quad
\begin{tabular}{c|cccc}
$A \land B$ & $\mathbf{t}$ & $\mathbf{b}$ & $\mathbf{n}$ & $\mathbf{f}$ \\ \hline
$\mathbf{t}$ & $\mathbf{t}$ & $\mathbf{b}$ & $\mathbf{n}$ & $\mathbf{f}$\\
$\mathbf{b}$ & $\mathbf{b}$ & $\mathbf{b}$ & $\mathbf{f}$ & $\mathbf{f}$\\
$\mathbf{n}$ & $\mathbf{n}$ & $\mathbf{f}$ & $\mathbf{n}$ & $\mathbf{f}$\\
$\mathbf{f}$ & $\mathbf{f}$ & $\mathbf{f}$ & $\mathbf{f}$ & $\mathbf{f}$\\
\end{tabular}
}
\end{displaymath}
Then, the semantic consequence relation for {\bf FDE}, $\models_{\bf FDE}$, is defined in terms of preservation of values ${\mathbf t}$ and ${\mathbf b}$ for all four-valued Belnap-Dunn-valuations. Moreover, the semantic consequence relations for {\bf NFL}, $\models_{\bf NFL}$, and {\bf ETL}, $\models_{\bf ETL}$, are defined by preserving values ${\mathbf t}$, ${\mathbf b}$ and ${\mathbf n}$ and the value ${\mathbf t}$, respectively, for all four-valued Belnap-Dunn-valuations.
\end{definition}

For the purpose of presenting an alternative semantics for {\bf NFL}, I make use of Routleys' invention.

\begin{definition}
A \emph{Routley interpretation} for $\mathcal{L}_{\bf FDE}$ is a structure $\langle W, g, \ast, v \rangle$ where $W$ is a set of worlds, $g\in W$, $\ast: W\longrightarrow W$ is a function with $w^{\ast\ast}=w$, and $v: W\times \mathsf{Prop}\longrightarrow \{ 0, 1 \}$. The function $v$ is extended to $I: W\times \mathsf{Form}_{\bf FDE}\longrightarrow \{ 0, 1 \}$ as follows:
\begin{itemize}
\setlength{\parskip}{0cm}
\setlength{\itemsep}{0cm}
\item $I(w, p)=v(w, p)$,
\item $I(w, \Not A)=1$ iff $I(w^\ast, A)\neq 1$,
\item $I(w, A\land B)=1$ iff $I(w, A)=1$ and $I(w, B)=1$,
\item $I(w, A\lor B)=1$ iff $I(w, A)=1$ or $I(w, B)=1$.
\end{itemize}
\end{definition}

Based on Routley interpretations, three consequence relations can be defined as follows.
\begin{definition}
For all $A, B\in \mathsf{Form}_{\bf FDE}$, 
\begin{itemize}
\setlength{\parskip}{0cm}
\setlength{\itemsep}{0cm}
\item $A \models_{\ast, \forall} B$ iff for all Routley interpretations $\langle W, g, \ast, v \rangle$, if $I(w, A)=1$ for all $w\in W$, then $I(w, B)=1$ for all $w\in W$.
\item $A \models_{\ast, g} B$ iff for all Routley interpretations $\langle W, g, \ast, v \rangle$, if $I(g, A){=}1$, then $I(g, B){=}1$.
\item $A \models_{\ast, \exists} B$ iff for all Routley interpretations $\langle W, g, \ast, v \rangle$, if $I(w, A)=1$ for some $w\in W$, then $I(w, B)=1$ for some $w\in W$.
\end{itemize}
\end{definition}

Then, the following results are obtained (the second item is due to Routleys).

\begin{theorem}
For all $A, B\in \mathsf{Form}_{\bf FDE}$, (i) $A \models_{\ast, \forall} B$ iff $A \models_{\bf ETL} B$; (ii) $A \models_{\ast, g} B$ iff $A \models_{\bf FDE} B$; (iii) $A \models_{\ast, \exists} B$ iff $A \models_{\bf NFL} B$.
\end{theorem}
\begin{proof}
The strategy is exactly the same as I did for the main result of the paper. I only note that for the first item, a Hilbert-style proof system introduced in \cite[\S3]{Pietz2013nothing} can be used. Therefore, I will only outline the case for the third item.

For the left-to-right direction, one should simply consider the Routley interpretations in which the cardinality of $W$ is two. Then, by unpacking the definition of Routley interpretations, $\models_{\bf NFL}$ is obtained. For the other way around, one may use of the proof system for {\bf NFL}, for example the one presented in \cite{Shramko2019NEBD}. Then, what remains to be done is to check the soundness, and this is tedious but not difficult. 
\end{proof}

\begin{remark}
In view of the recent revival of $p$- and $q$-consequence relations (cf. \cite{Malinowski1990Q,Frankowski2004formalization,Frankowski2004p}), through a series of papers by Pablo Cobreros, Paul Egr\'e, Dave Ripley, and Robert van Rooij (e.g. \cite{Cobreros2012tolerant,Cobreros2013reaching}), the above result seems to imply that Ja\'skowski's idea can be exported to enrich the $p$- and $q$-consequence relations by modal vocabularies that are characterized in terms of Kripke models. Whether this is the case, and if so then how this might be developed remains to be seen, and is left as a topic for further investigations.
\end{remark}


\section{Concluding remarks}
Discussive logics are often characterized as typical paraconsistent logics in which the rule of adjunction fails. The failure of adjunction is of course true for the non-discussive conjunction, but false for discussive conjunction. In fact, the negation-free fragment of $\mathcal{L}_r$ and $\mathcal{L}_l$ are both completely classical. 

What I hope to have pointed out, as an application of the main result, is an aspect of discussive logics beyond the failure of adjunction. More specifically, it seems that the discussive semantics can be seen as a tool to make sense of certain many-valued semantics that may look rather difficult to have an intuitive grasp of. The key feature of the discussive semantics is this: just require one of the points in the model to force formulas in order to define the validity. Of course, the rule of adjunction will fail for non-discussive conjunction because of this key feature. But, its effect goes well beyond the failure of adjunction since one may consider discussive semantics for languages without conjunction, such as the negation-conditional fragment of {\bf D$_2$}. It therefore seems that there is more to discussive logics than the failure of adjunction. 

Finally, building on this view of discussive logics, there seem to be a number of future directions. For instance, thanks to the simplicity of the key feature, discussive variants can be considered for a wide range of logics with Kripke models. A systematic investigation of this question from both technical as well as philosophical perspective remains to be seen. For the former, a first step is marked by Lloyd Humberstone in \cite{Humberstone2008modal}. For the latter, the discussion by Priest on Jaina logic in \cite{Priest2008jaina} seems to be promising, beside the topics related to $p$- and $q$-consequence relations mentioned above.











\begin{thebibliography}{10}
	\providecommand{\bibitemdeclare}[2]{}
	\providecommand{\surnamestart}{}
	\providecommand{\surnameend}{}
	\providecommand{\urlprefix}{Available at }
	\providecommand{\url}[1]{\texttt{#1}}
	\providecommand{\href}[2]{\texttt{#2}}
	\providecommand{\urlalt}[2]{\href{#1}{#2}}
	\providecommand{\doi}[1]{doi:\urlalt{https://doi.org/#1}{#1}}
	\providecommand{\eprint}[1]{arXiv:\urlalt{https://arxiv.org/abs/#1}{#1}}
	\providecommand{\bibinfo}[2]{#2}
	
	\bibitemdeclare{article}{AAZ2011}
	\bibitem{AAZ2011}
	\bibinfo{author}{Ofer \surnamestart Arieli\surnameend}, \bibinfo{author}{Arnon
		\surnamestart Avron\surnameend} \& \bibinfo{author}{Anna \surnamestart
		Zamansky\surnameend} (\bibinfo{year}{2011}): \emph{\bibinfo{title}{Maximal
			and Premaximal Paraconsistency in the Framework of Three-Valued Semantics}}.
	\newblock {\slshape \bibinfo{journal}{Studia Logica}} \bibinfo{volume}{97}, pp.
	\bibinfo{pages}{31--60}, \doi{10.1007/s11225-010-9296-9}.
	
	\bibitemdeclare{article}{Avron1986}
	\bibitem{Avron1986}
	\bibinfo{author}{Arnon \surnamestart Avron\surnameend} (\bibinfo{year}{1986}):
	\emph{\bibinfo{title}{{On An Implication Connective of RM}}}.
	\newblock {\slshape \bibinfo{journal}{Notre Dame Journal of Formal Logic}}
	\bibinfo{volume}{27}(\bibinfo{number}{2}), pp. \bibinfo{pages}{201--209},
	\doi{10.1305/NDJFL/1093636612}.
	
	\bibitemdeclare{article}{BatDeCl2004}
	\bibitem{BatDeCl2004}
	\bibinfo{author}{Diderik \surnamestart Batens\surnameend} \&
	\bibinfo{author}{Kristof \surnamestart {De Clercq}\surnameend}
	(\bibinfo{year}{2004}): \emph{\bibinfo{title}{A Rich Paraconsistent Extension
			of Full Positive Logic}}.
	\newblock {\slshape \bibinfo{journal}{Logique et Analyse}}
	\bibinfo{volume}{185-188}, pp. \bibinfo{pages}{227--257}.
	
	\bibitemdeclare{article}{CMdA2000}
	\bibitem{CMdA2000}
	\bibinfo{author}{Walter \surnamestart Carnielli\surnameend},
	\bibinfo{author}{Joao \surnamestart Marcos\surnameend} \&
	\bibinfo{author}{Sandra \surnamestart de~Amo\surnameend}
	(\bibinfo{year}{2000}): \emph{\bibinfo{title}{Formal Inconsistency and
			Evolutionary Databases}}.
	\newblock {\slshape \bibinfo{journal}{Logic and Logical Philosophy}}
	\bibinfo{volume}{8}, pp. \bibinfo{pages}{115--152},
	\doi{10.12775/LLP.2000.008.}
	
	\bibitemdeclare{incollection}{Carnielli1999society}
	\bibitem{Carnielli1999society}
	\bibinfo{author}{Walter~A \surnamestart Carnielli\surnameend} \&
	\bibinfo{author}{Mamede \surnamestart Lima-Marques\surnameend}
	(\bibinfo{year}{1999}): \emph{\bibinfo{title}{Society semantics and
			multiple-valued logics}}.
	\newblock In: {\slshape \bibinfo{booktitle}{Contemporary Mathematics}},
	\bibinfo{volume}{235}, \bibinfo{publisher}{American Mathematical Society},
	pp. \bibinfo{pages}{33--52}, \doi{10.1090/conm/235/03464}.
	
	\bibitemdeclare{article}{C1}
	\bibitem{C1}
	\bibinfo{author}{Janusz \surnamestart Ciuciura\surnameend}
	(\bibinfo{year}{2006}): \emph{\bibinfo{title}{{A Quasi-Discursive System {\it
					ND$_2^+$}}}}.
	\newblock {\slshape \bibinfo{journal}{Notre Dame Journal of Formal Logic}}
	\bibinfo{volume}{47}, pp. \bibinfo{pages}{371--384},
	\doi{10.1305/ndjfl/1163775444}.
	
	\bibitemdeclare{article}{ciuciura2008frontiers}
	\bibitem{ciuciura2008frontiers}
	\bibinfo{author}{Janusz \surnamestart Ciuciura\surnameend}
	(\bibinfo{year}{2008}): \emph{\bibinfo{title}{Frontiers of the discursive
			logic}}.
	\newblock {\slshape \bibinfo{journal}{Bulletin of the Section of Logic}}
	\bibinfo{volume}{37}(\bibinfo{number}{2}), pp. \bibinfo{pages}{81--92}.
	
	\bibitemdeclare{article}{Cobreros2012tolerant}
	\bibitem{Cobreros2012tolerant}
	\bibinfo{author}{Pablo \surnamestart Cobreros\surnameend},
	\bibinfo{author}{Paul \surnamestart Egr{\'e}\surnameend},
	\bibinfo{author}{David \surnamestart Ripley\surnameend} \&
	\bibinfo{author}{Robert \surnamestart van Rooij\surnameend}
	(\bibinfo{year}{2012}): \emph{\bibinfo{title}{Tolerant, classical, strict}}.
	\newblock {\slshape \bibinfo{journal}{Journal of Philosophical Logic}}
	\bibinfo{volume}{41}(\bibinfo{number}{2}), pp. \bibinfo{pages}{347--385},
	\doi{10.1007/s10992-010-9165-z}.
	
	\bibitemdeclare{article}{Cobreros2013reaching}
	\bibitem{Cobreros2013reaching}
	\bibinfo{author}{Pablo \surnamestart Cobreros\surnameend},
	\bibinfo{author}{Paul \surnamestart {\'E}gr{\'e}\surnameend},
	\bibinfo{author}{David \surnamestart Ripley\surnameend} \&
	\bibinfo{author}{Robert \surnamestart Van~Rooij\surnameend}
	(\bibinfo{year}{2013}): \emph{\bibinfo{title}{Reaching transparent truth}}.
	\newblock {\slshape \bibinfo{journal}{Mind}}
	\bibinfo{volume}{122}(\bibinfo{number}{488}), pp. \bibinfo{pages}{841--866},
	\doi{10.1093/mind/fzt110}.
	
	\bibitemdeclare{incollection}{dacosta1977new}
	\bibitem{dacosta1977new}
	\bibinfo{author}{Newton C.~A. \surnamestart da~Costa\surnameend} \&
	\bibinfo{author}{Lech \surnamestart Dubikajtis\surnameend}
	(\bibinfo{year}{1977}): \emph{\bibinfo{title}{{On Ja\'skowski's Discussive
				Logic}}}.
	\newblock In \bibinfo{editor}{A.~I. \surnamestart Arruda\surnameend},
	\bibinfo{editor}{N.~C.~A. \surnamestart da~Costa\surnameend} \&
	\bibinfo{editor}{R.~\surnamestart Chuaqui\surnameend}, editors: {\slshape
		\bibinfo{booktitle}{Non Classical Logics, Model Theory and Computability}},
	\bibinfo{publisher}{North-Holland}, pp. \bibinfo{pages}{37--56},
	\doi{10.1016/S0049-237X(08)70644-X}.
	
	\bibitemdeclare{article}{Dunn1976}
	\bibitem{Dunn1976}
	\bibinfo{author}{Michael \surnamestart Dunn\surnameend} (\bibinfo{year}{1976}):
	\emph{\bibinfo{title}{Intuitive semantics for first-degree entailments and
			`coupled trees'}}.
	\newblock {\slshape \bibinfo{journal}{Philosophical studies}}
	\bibinfo{volume}{29}(\bibinfo{number}{3}), pp. \bibinfo{pages}{149--168},
	\doi{10.1007/BF00373152}.
	
	\bibitemdeclare{article}{Frankowski2004formalization}
	\bibitem{Frankowski2004formalization}
	\bibinfo{author}{Szymon \surnamestart Frankowski\surnameend}
	(\bibinfo{year}{2004}): \emph{\bibinfo{title}{Formalization of a plausible
			inference}}.
	\newblock {\slshape \bibinfo{journal}{Bulletin of the Section of Logic}}
	\bibinfo{volume}{33}(\bibinfo{number}{1}), pp. \bibinfo{pages}{41--52}.
	
	\bibitemdeclare{article}{Frankowski2004p}
	\bibitem{Frankowski2004p}
	\bibinfo{author}{Szymon \surnamestart Frankowski\surnameend}
	(\bibinfo{year}{2004}): \emph{\bibinfo{title}{{$p$-consequence Versus
				$q$-consequence Operations}}}.
	\newblock {\slshape \bibinfo{journal}{Bulletin of the Section of Logic}}
	\bibinfo{volume}{33}(\bibinfo{number}{4}), pp. \bibinfo{pages}{197--207}.
	
	\bibitemdeclare{book}{Haack1978}
	\bibitem{Haack1978}
	\bibinfo{author}{Susan \surnamestart Haack\surnameend} (\bibinfo{year}{1978}):
	\emph{\bibinfo{title}{Philosophy of Logics}}.
	\newblock \bibinfo{publisher}{Cambridge University Press},
	\doi{10.1017/CBO9780511812866}.
	
	\bibitemdeclare{article}{herzberger1973dimensions}
	\bibitem{herzberger1973dimensions}
	\bibinfo{author}{Hans~G. \surnamestart Herzberger\surnameend}
	(\bibinfo{year}{1973}): \emph{\bibinfo{title}{Dimensions of truth}}.
	\newblock {\slshape \bibinfo{journal}{Journal of Philosophical Logic}}
	\bibinfo{volume}{2}(\bibinfo{number}{4}), pp. \bibinfo{pages}{535--556},
	\doi{10.1007/bf00262954}.
	
	\bibitemdeclare{article}{Humberstone2008modal}
	\bibitem{Humberstone2008modal}
	\bibinfo{author}{Lloyd \surnamestart Humberstone\surnameend}
	(\bibinfo{year}{2008}): \emph{\bibinfo{title}{Modal formulas true at some
			point in every model}}.
	\newblock {\slshape \bibinfo{journal}{The Australasian Journal of Logic}}
	\bibinfo{volume}{6}, pp. \bibinfo{pages}{70--82},
	\doi{10.26686/ajl.v6i0.1794}.
	
	\bibitemdeclare{article}{J1'}
	\bibitem{J1'}
	\bibinfo{author}{Stanis{\l}aw \surnamestart Ja\'skowski\surnameend}
	(\bibinfo{year}{1969}): \emph{\bibinfo{title}{{Propositional Calculus for
				Contradictory Deductive Systems}}}.
	\newblock {\slshape \bibinfo{journal}{Studia Logica}} \bibinfo{volume}{24}, pp.
	\bibinfo{pages}{143--157}, \doi{10.1007/BF02134311}.
	
	\bibitemdeclare{article}{J1}
	\bibitem{J1}
	\bibinfo{author}{Stanis{\l}aw \surnamestart Ja\'skowski\surnameend}
	(\bibinfo{year}{1999}): \emph{\bibinfo{title}{{A Propositional Calculus for
				Inconsistent Deductive Systems}}}.
	\newblock {\slshape \bibinfo{journal}{Logic and Logical Philosophy}}
	\bibinfo{volume}{7}, pp. \bibinfo{pages}{35--56},
	\doi{10.12775/LLP.1999.003}.
	\newblock \bibinfo{note}{A new translation of \cite{J1'}.}
	
	\bibitemdeclare{article}{J2}
	\bibitem{J2}
	\bibinfo{author}{Stanis{\l}aw \surnamestart Ja\'skowski\surnameend}
	(\bibinfo{year}{1999}): \emph{\bibinfo{title}{{On the Discussive Conjunction
				in the Propositional Calculus for Inconsistent Deductive Systems}}}.
	\newblock {\slshape \bibinfo{journal}{Logic and Logical Philosophy}}
	\bibinfo{volume}{7}, pp. \bibinfo{pages}{57--59},
	\doi{10.12775/LLP.1999.004}.
	
	\bibitemdeclare{article}{Kotas1974b}
	\bibitem{Kotas1974b}
	\bibinfo{author}{Jerzy \surnamestart Kotas\surnameend} (\bibinfo{year}{1974}):
	\emph{\bibinfo{title}{{On Quantity of Logical Values in the Discussive D$_2$
				System and in Modular Logic}}}.
	\newblock {\slshape \bibinfo{journal}{Studia Logica}}
	\bibinfo{volume}{33}(\bibinfo{number}{3}), pp. \bibinfo{pages}{273--275},
	\doi{doi.org/10.1007/BF02123281}.
	
	\bibitemdeclare{article}{Malinowski1990Q}
	\bibitem{Malinowski1990Q}
	\bibinfo{author}{Grzegorz \surnamestart Malinowski\surnameend}
	(\bibinfo{year}{1990}): \emph{\bibinfo{title}{{$Q$-consequence operation}}}.
	\newblock {\slshape \bibinfo{journal}{Reports on Mathematical Logic}}
	\bibinfo{volume}{24}, pp. \bibinfo{pages}{49--59}.
	
	\bibitemdeclare{incollection}{Marcos2011}
	\bibitem{Marcos2011}
	\bibinfo{author}{Jo\~{a}o \surnamestart Marcos\surnameend}
	(\bibinfo{year}{2011}): \emph{\bibinfo{title}{The value of the two values}}.
	\newblock In \bibinfo{editor}{J.-Y. \surnamestart B\'eziau\surnameend} \&
	\bibinfo{editor}{M.E. \surnamestart Coniglio\surnameend}, editors: {\slshape
		\bibinfo{booktitle}{Logic without Frontiers: Festschrift for Walter Alexandre
			Carnielli on the occasion of his 60th birthday}}, \bibinfo{publisher}{College
		Publication}, pp. \bibinfo{pages}{277--294}.
	
	\bibitemdeclare{inproceedings}{OmoriLORI2017}
	\bibitem{OmoriLORI2017}
	\bibinfo{author}{Hitoshi \surnamestart Omori\surnameend}
	(\bibinfo{year}{2017}): \emph{\bibinfo{title}{Sette's Logics, Revisited}}.
	\newblock In \bibinfo{editor}{Alexandru \surnamestart Baltag\surnameend},
	\bibinfo{editor}{Jeremy \surnamestart Seligman\surnameend} \&
	\bibinfo{editor}{Tomoyuki \surnamestart Yamada\surnameend}, editors:
	{\slshape \bibinfo{booktitle}{Proceedings of {LORI} 2017}}, pp.
	\bibinfo{pages}{451--465}, \doi{10.1007/978-3-662-55665-8{\_}31}.
	
	\bibitemdeclare{article}{OAD2}
	\bibitem{OAD2}
	\bibinfo{author}{Hitoshi \surnamestart Omori\surnameend} \&
	\bibinfo{author}{Jesse \surnamestart Alama\surnameend}
	(\bibinfo{year}{2018}): \emph{\bibinfo{title}{{Axiomatizing Ja\'skowski's
				Discussive Logic {\bf D$_2$}}}}.
	\newblock {\slshape \bibinfo{journal}{Studia Logica}}
	\bibinfo{volume}{106}(\bibinfo{number}{6}), pp. \bibinfo{pages}{1163--1180},
	\doi{10.1007/s11225-017-9780-6}.
	
	\bibitemdeclare{inproceedings}{OmoriISMVL22}
	\bibitem{OmoriISMVL22}
	\bibinfo{author}{Hitoshi \surnamestart Omori\surnameend} \&
	\bibinfo{author}{Jonas R.~B. \surnamestart Arenhart\surnameend}
	(\bibinfo{year}{2022}): \emph{\bibinfo{title}{{Haack meets Herzberger and
				Priest}}}.
	\newblock In: {\slshape \bibinfo{booktitle}{2022 IEEE 52nd International
			Symposium on Multiple-Valued Logic (ISMVL)}}, pp. \bibinfo{pages}{137--144},
	\doi{10.1109/ISMVL52857.2022.00028}.
	
	\bibitemdeclare{article}{Omori2023-theoria}
	\bibitem{Omori2023-theoria}
	\bibinfo{author}{Hitoshi \surnamestart Omori\surnameend} \&
	\bibinfo{author}{Jonas R.~B. \surnamestart Arenhart\surnameend}
	(\bibinfo{year}{2023}): \emph{\bibinfo{title}{Change of logic, without change
			of meaning}}.
	\newblock {\slshape \bibinfo{journal}{Theoria}}
	\bibinfo{volume}{89}(\bibinfo{number}{4}), pp. \bibinfo{pages}{414--431},
	\doi{10.1111/theo.12459}.
	
	\bibitemdeclare{article}{Omori2023finetti}
	\bibitem{Omori2023finetti}
	\bibinfo{author}{Hitoshi \surnamestart Omori\surnameend} \&
	\bibinfo{author}{Jonas R.~B. \surnamestart Arenhart\surnameend}
	(\bibinfo{year}{2024}): \emph{\bibinfo{title}{Is the de Finetti conditional a
			conditional?}}
	\newblock {\slshape \bibinfo{journal}{Argumenta}},
	\doi{10.14275/2465-2334/20230.omo}.
	
	\bibitemdeclare{article}{Pietz2013nothing}
	\bibitem{Pietz2013nothing}
	\bibinfo{author}{Andreas \surnamestart Pietz\surnameend} \&
	\bibinfo{author}{Umberto \surnamestart Rivieccio\surnameend}
	(\bibinfo{year}{2013}): \emph{\bibinfo{title}{Nothing but the truth}}.
	\newblock {\slshape \bibinfo{journal}{Journal of Philosophical Logic}}
	\bibinfo{volume}{42}(\bibinfo{number}{1}), pp. \bibinfo{pages}{125--135},
	\doi{10.1007/s10992-011-9215-1}.
	
	\bibitemdeclare{article}{Priest79}
	\bibitem{Priest79}
	\bibinfo{author}{Graham \surnamestart Priest\surnameend}
	(\bibinfo{year}{1979}): \emph{\bibinfo{title}{{The Logic of Paradox}}}.
	\newblock {\slshape \bibinfo{journal}{Journal of Philosophical Logic}}
	\bibinfo{volume}{8}(\bibinfo{number}{1}), pp. \bibinfo{pages}{219--241},
	\doi{10.1007/BF00258428}.
	
	\bibitemdeclare{article}{Priest1984hyper}
	\bibitem{Priest1984hyper}
	\bibinfo{author}{Graham \surnamestart Priest\surnameend}
	(\bibinfo{year}{1984}): \emph{\bibinfo{title}{Hyper-contradictions}}.
	\newblock {\slshape \bibinfo{journal}{Logique et Analyse}}
	\bibinfo{volume}{27}(\bibinfo{number}{107}), pp. \bibinfo{pages}{237--243}.
	
	\bibitemdeclare{article}{Priest2008jaina}
	\bibitem{Priest2008jaina}
	\bibinfo{author}{Graham \surnamestart Priest\surnameend}
	(\bibinfo{year}{2008}): \emph{\bibinfo{title}{{Jaina logic: A contemporary
				perspective}}}.
	\newblock {\slshape \bibinfo{journal}{History and Philosophy of Logic}}
	\bibinfo{volume}{29}(\bibinfo{number}{3}), pp. \bibinfo{pages}{263--278},
	\doi{10.1080/01445340701690233}.
	
	\bibitemdeclare{article}{Priest2014}
	\bibitem{Priest2014}
	\bibinfo{author}{Graham \surnamestart Priest\surnameend}
	(\bibinfo{year}{2014}): \emph{\bibinfo{title}{Plurivalent {L}ogics}}.
	\newblock {\slshape \bibinfo{journal}{The Australasian Journal of Logic}}
	\bibinfo{volume}{11}(\bibinfo{number}{1}), pp. \bibinfo{pages}{1--13},
	\doi{10.26686/ajl.v11i1.1830}.
	
	\bibitemdeclare{article}{Routleys1972}
	\bibitem{Routleys1972}
	\bibinfo{author}{Richard \surnamestart Routley\surnameend} \&
	\bibinfo{author}{Valerie \surnamestart Routley\surnameend}
	(\bibinfo{year}{1972}): \emph{\bibinfo{title}{Semantics for first degree
			entailment}}.
	\newblock {\slshape \bibinfo{journal}{No{\^u}s}} \bibinfo{volume}{6}, pp.
	\bibinfo{pages}{335--359}, \doi{10.2307/2214309}.
	
	\bibitemdeclare{article}{Schumm1975henkin}
	\bibitem{Schumm1975henkin}
	\bibinfo{author}{George~F. \surnamestart Schumm\surnameend}
	(\bibinfo{year}{1975}): \emph{\bibinfo{title}{{A Henkin-style completeness
				proof for the pure implicational calculus}}}.
	\newblock {\slshape \bibinfo{journal}{Notre Dame Journal of Formal Logic}}
	\bibinfo{volume}{16}(\bibinfo{number}{3}), pp. \bibinfo{pages}{402--404},
	\doi{10.1305/ndjfl/1093891803}.
	
	\bibitemdeclare{article}{Sette1973}
	\bibitem{Sette1973}
	\bibinfo{author}{Antonio \surnamestart Sette\surnameend}
	(\bibinfo{year}{1973}): \emph{\bibinfo{title}{On the propositional calculus
			{P}$^1$}}.
	\newblock {\slshape \bibinfo{journal}{Mathematica Japonicae}}
	\bibinfo{volume}{18}(\bibinfo{number}{3}), pp. \bibinfo{pages}{173--180}.
	
	\bibitemdeclare{incollection}{Shramko2019NEBD}
	\bibitem{Shramko2019NEBD}
	\bibinfo{author}{Yaroslav \surnamestart Shramko\surnameend}
	(\bibinfo{year}{2019}): \emph{\bibinfo{title}{First-Degree Entailment and
			Structural Reasoning}}.
	\newblock In \bibinfo{editor}{Hitoshi \surnamestart Omori\surnameend} \&
	\bibinfo{editor}{Heinrich \surnamestart Wansing\surnameend}, editors:
	{\slshape \bibinfo{booktitle}{{New Essays on Belnap-Dunn Logic}}}, {\slshape
		\bibinfo{series}{Synthese Library}} \bibinfo{volume}{418},
	\bibinfo{publisher}{Springer}, pp. \bibinfo{pages}{311--324},
	\doi{10.1007/978-3-030-31136-0{\_}18}.
	
	\bibitemdeclare{article}{Shramko2017}
	\bibitem{Shramko2017}
	\bibinfo{author}{Yaroslav \surnamestart Shramko\surnameend},
	\bibinfo{author}{Dmitry \surnamestart Zaitsev\surnameend} \&
	\bibinfo{author}{Alexander \surnamestart Belikov\surnameend}
	(\bibinfo{year}{2017}): \emph{\bibinfo{title}{First-Degree Entailment and its
			Relatives}}.
	\newblock {\slshape \bibinfo{journal}{Studia Logica}}
	\bibinfo{volume}{105}(\bibinfo{number}{6}), pp. \bibinfo{pages}{1291--1317},
	\doi{10.1007/s11225-017-9747-7}.
	
	\bibitemdeclare{article}{Vasyukov}
	\bibitem{Vasyukov}
	\bibinfo{author}{Vladimir~L. \surnamestart Vasyukov\surnameend}
	(\bibinfo{year}{2001}): \emph{\bibinfo{title}{{A New Axiomatization of
				Ja\'skowski's Discussive Logic}}}.
	\newblock {\slshape \bibinfo{journal}{Logic and Logical Philosophy}}
	\bibinfo{volume}{9}, pp. \bibinfo{pages}{35--46},
	\doi{10.12775/LLP.2001.003}.
	
\end{thebibliography}

\end{document}